\documentclass[a4paper,10pt]{article}

\usepackage[utf8]{inputenc}     

\usepackage{amsfonts}
\usepackage{amsthm}
\usepackage{amssymb}
\usepackage{amsmath}
\usepackage{tikz}
\usetikzlibrary{arrows,automata}
\usepackage{dsfont}
\usepackage{hyperref} 

\newtheorem{theorem}{Theorem}

\newtheorem{lemma}[theorem]{Lemma}
\newtheorem{proposition}[theorem]{Proposition}

\newcommand{\tuple}[1]{\langle #1 \rangle}
\newcommand{\emptyword}[0]{\varepsilon}
\newcommand{\trans}[1]{\mathchoice{\xrightarrow{#1}}{\xrightarrow{\smash{\lower1pt\hbox{$\scriptstyle #1$}}}}{\xrightarrow{#1}}{\xrightarrow{#1}}}
\newcommand{\future}[1]{\mathsf{F}(#1)}
\newcommand{\clofut}[1]{\overline{\mathsf{F}}(#1)}
\newcommand{\bifut}[1]{\mathsf{F}_2(#1)}
\newcommand{\clobifut}[1]{\overline{\mathsf{F}}_2(#1)}
\newcommand{\past}[1]{\mathsf{P}(#1)}
\newcommand{\fact}[1]{\operatorname{fact}(#1)}
\newcommand{\supp}[1]{\operatorname{supp}(#1)}

\DeclareUnicodeCharacter{D7}{\times}            
\DeclareUnicodeCharacter{394}{\Delta}           
\DeclareUnicodeCharacter{3A0}{\Pi}              
\DeclareUnicodeCharacter{3B1}{\alpha}           
\DeclareUnicodeCharacter{3B2}{\beta}            
\DeclareUnicodeCharacter{3B4}{\delta}           
\DeclareUnicodeCharacter{3B5}{\varepsilon}      
\DeclareUnicodeCharacter{3BB}{\lambda}          
\DeclareUnicodeCharacter{3BC}{\mu}              
\DeclareUnicodeCharacter{3BD}{\nu}              
\DeclareUnicodeCharacter{3C0}{\pi}              
\DeclareUnicodeCharacter{3C1}{\rho}             
\DeclareUnicodeCharacter{3C3}{\sigma}           
\DeclareUnicodeCharacter{2026}{\ldots}          
\DeclareUnicodeCharacter{2113}{\ell}            
\DeclareUnicodeCharacter{2115}{\mathbb{N}}      
\DeclareUnicodeCharacter{212C}{\mathcal{B}}     
\DeclareUnicodeCharacter{2192}{\rightarrow}     
\DeclareUnicodeCharacter{2203}{\exists}         
\DeclareUnicodeCharacter{2205}{\varnothing}     
\DeclareUnicodeCharacter{2208}{\in}             
\DeclareUnicodeCharacter{2209}{\notin}          
\DeclareUnicodeCharacter{2211}{\sum}            
\DeclareUnicodeCharacter{2216}{\setminus}       
\DeclareUnicodeCharacter{221E}{\infty}          
\DeclareUnicodeCharacter{2227}{\wedge}          
\DeclareUnicodeCharacter{2229}{\cap}            
\DeclareUnicodeCharacter{222A}{\cup}            
\DeclareUnicodeCharacter{223C}{\sim}            
\DeclareUnicodeCharacter{2260}{\neq}            
\DeclareUnicodeCharacter{2283}{\supset}         
\DeclareUnicodeCharacter{2286}{\subseteq}       
\DeclareUnicodeCharacter{2287}{\supseteq}       
\DeclareUnicodeCharacter{228E}{\uplus}          
\DeclareUnicodeCharacter{22C3}{\bigcup}         
\DeclareUnicodeCharacter{22C5}{\cdot}           
\DeclareUnicodeCharacter{22EF}{\cdots}          
\DeclareUnicodeCharacter{27E8}{\langle}         
\DeclareUnicodeCharacter{27E9}{\rangle}         
\DeclareUnicodeCharacter{2A04}{\biguplus}       
\DeclareUnicodeCharacter{2A7D}{\leqslant}       
\DeclareUnicodeCharacter{2A7E}{\geqslant}       
\DeclareUnicodeCharacter{1D49C}{\mathcal{A}}    
\DeclareUnicodeCharacter{1D4A2}{\mathcal{G}}    
\DeclareUnicodeCharacter{1D4AF}{\mathcal{T}}    
\DeclareUnicodeCharacter{1D7D9}{\mathds{1}}     

\begin{document}

\title{Ambiguity through the lens of measure theory}
\author{Olivier Carton}
\date{\today}
\maketitle

\begin{abstract}
  In this paper, we establish a strong link between the ambiguity for
  finite words of a Büchi automaton and the ambiguity for infinite words of
  the same automaton.  This link is based on measure theory.  More
  precisely, we show that such an automaton is unambiguous, in the sense
  that no finite word labels two runs with the same starting state and the
  same ending state if and only if for each state, the set of infinite
  sequences labelling two runs starting from that state has measure zero.
  The measure used to define these negligible sets, that is sets of measure
  zero, can be any measure computed by a weighted automaton which is
  compatible with the Büchi automaton.  This latter condition is very
  natural: the measure must only put weight on sets $wA^ℕ$ where $w$
  is the label of some run in the Büchi automaton.
\end{abstract}

\section{Introduction}

The relationship between deterministic and non-deterministic machines has
been extensively studied since the very beginning of computer
science. Despite these efforts, many questions remain wide open.  This is
of course true in complexity theory for questions like P versus NP but also
in automata theory \cite{HolzerKutrib10,Colcombet15}.  It is for instance
not known whether the simulation of non-deterministic either one-way or
two-way automata by deterministic two-way automata requires an exponential
blow-up of the number of states \cite{Pighizzini12}.

Unambiguous machines are usually defined as non-deterministic machines in
which each input has at most one accepting run.  They are intermediate
machines in between the two extreme cases of deterministic and
non-deterministic machines. The notion of ambiguity considered in the paper
is slightly stronger and more structural as it does not depend on initial
and final states.  In the case of automata accepting finite words,
non-deterministic automata can be exponentially more succinct than
unambiguous automata which can be, in turn, exponentially more succinct
than deterministic automata \cite{Schmidt78}.  However, the problem of
containment for unambiguous automata is tractable is polynomial time
\cite{StearnsHunt85} like for deterministic automata while the same problem
for non-deterministic automata is PSPACE-complete
\cite[Section~10.6]{AhoHopcroftUllman74}.

The polynomial time algorithm for the containment of unambiguous automata
accepting finite words in~\cite{StearnsHunt85} is based on a clever
counting argument which cannot easily be adapted to infinite words. It is
still unknown whether the containment problem for unambiguous Büchi
automata can be solved in polynomial time.  The problem was solved
in~\cite{IsaakLoeding12} for sub-classes of Büchi automata with weak
acceptance conditions and in~\cite{BousquetLoeding10} for prophetic Büchi
automata introduced in \cite{CartonMichel03} (see also
\cite[Sec.~II.10]{PerrinPin04}) which are strongly unambiguous.  These
latter results are obtained through reductions of the problem for infinite
words to the problem for finite words.  The main result of this paper can
be seen as a step towards a solution for all Büchi automata as it connects
ambiguity for infinite words to ambiguity for finite words.

The aim of this paper is to exhibit a strong link between the ambiguity of
some automaton for finite words and the ambiguity of the same automaton for
infinite words.  The paper is focused on strongly connected Büchi automata.
Two examples given in the conclusion show that the problem is more involved
for non strongly connected automata.  It turns out that unambiguity for
infinite words implies the unambiguity for finite words but the converse
does not hold in general.  This converse can however be recovered if
unambiguity for infinite words is considered up to a negligible set of
inputs.  \emph{Negligible} should here be understood as a set of zero
measure.  The measure used to characterize ambiguity must fulfill some
compatibility conditions with the automaton.  Some examples show these
conditions cannot be avoided.  Note that measures were already used to
characterize maximal variable-length codes
\cite[Thm.~5.10]{BerstelPerrin84} which are, in essence, a combinatorial
definition of non-ambiguity.

The first step of the proof is to show that the measure of the set of
accepted sequences does not increase if all states of the automaton are
made final.  This result is interesting by itself but it also reduces the
proof of our result to automata with all states final.  Since initial
states are also not relevant, the problem is again reduced to automata with
all states initial and final, which accept the so called shift spaces from
symbolic dynamics \cite{LindMarcus95}.  This special case is handled using
techniques from this domain like synchronizing words and Fisher covers.

This work was motivated by questions about automata with outputs also known
as transducers.  These transducers realize functions mapping infinite
sequences to infinite sequences and the questions are focused on the long
term behaviour.  A natural question is the preservation of normality where
normality is the property that all blocks of the same length occur with the
same limiting frequency~\cite{Carton22}.  Normality was introduced by Borel
to formalize the most basic form of randomness for real
numbers~\cite{Borel09}.  It turns out that normality can be characterized
by non-compressibility by transducers realizing one-to-one
functions~\cite{BecherCarton18}.  Since each infinite run ends in a
strongly connected component, it is sufficient to study ambiguity of
strongly connected automata.  It is a classical result that each function
realized by a transducer can be realized by a transducer whose input
automaton is unambiguous~\cite{ChoffrutGrigorieff99}.  The result proved in
this paper shows that if all states of a strongly connected unambiguous
transducer are made final the transducer remains unambiguous up to a set of
measure zero.  It allows us to use, for instance, the ergodic theorem for
Markov chains where the function must be defined up to a set of measure
zero.

The paper is organized as follows.  Section~\ref{sec:basic} is devoted to
basic definitions needed for the main result which is stated in
Section~\ref{sec:result}.  The first step of the proof is to reduce the
problem to the special case of Büchi automata with all states final.  This
is done in Section~\ref{sec:reduction}.  The proof of this special case is
carried out in Section~\ref{sec:proof}.

\section{Definitions} \label{sec:basic}

\subsection{Words, sequences and measures}

Let $A$ be a finite set of \emph{symbols} that we refer to as the
\emph{alphabet}. We write $A^ℕ$ for the set of all sequences on the
alphabet~$A$ and $A^*$ for the set of all (finite) words.  The length of a
finite word $w$ is denoted by $|w|$.  The positions of sequences and words
are numbered starting from~$1$.  The empty word is denoted by~$\emptyword$.
The cardinality of a finite set~$E$ is denoted by~$\#E$.  A \emph{factor}
of a sequence $a_1a_2a_3 ⋯$ is a finite word of the form $a_ka_{k+1}⋯
a_{ℓ-1}$ for integers $1 ⩽ k ⩽ ℓ$ where $k = ℓ$ yields the empty
word~$\emptyword$.  We let $\fact{X}$ denote the set of factors of a
set~$X$ of sequences.

We recall here a few notions of topology.  The set~$A^ℕ$ of sequences can be
endowed with a topology by the distance~$d$ which is defined as follows.
The distance $d(x,y)$ of two sequences $x = a_1a_2a_3⋯$ and $y =
b_1b_2b_3⋯$ is zero if $x = y$ and is $2^{-\min\{ i: a_i ≠ b_i\}}$
otherwise.  The set~$X$ is \emph{open} if it is equal to a possibly
infinite union of \emph{cylinders}, that is, sets of the form~$wA^ℕ$ for $w
∈ A^*$.  It is \emph{closed} if its complement in~$A^ℕ$ is open.  Each
set~$X$ is contained in a smallest closed set~$\overline{X}$ called its
\emph{closure}.  The complement of~$\overline{X}$ is the union of all sets
$wA^ℕ$ which are disjoint from~$X$.

We present here the key notion of a measure.  It is seen as a function from
finite to real numbers in $[0; 1]$ which assigns a measure to each cylinder
set $wA^ℕ$.  Then it is extended as a measure of sets of infinite sequences
by Carathéodory extension theorem. A \emph{probability measure on $A^*$} is
a function $μ : A^* → [0,1]$ such that $μ(\emptyword) = 1$ and that the
equality
\begin{displaymath}
  ∑_{a ∈ A}{μ(wa)} = μ(w)
\end{displaymath}
holds for each word $w ∈ A^*$.  The simplest example of a probability
measure is a \emph{Bernoulli measure}.  It is a monoid morphism from~$A^*$
to~$[0,1]$ (endowed with multiplication) such that
$∑_{a ∈ A}{μ(a)} = 1$.  Among the Bernoulli measures is the
\emph{uniform measure} which maps each word $w ∈ A^*$ to $(\#A)^{-|w|}$.
In particular, each symbol~$a$ is mapped to $μ(a) = 1/\#A$.

By the Carathéodory extension theorem, a measure~$μ$ on~$A^*$ can be
uniquely extended to a probability measure~$\hat{μ}$ on~$A^ℕ$ such that
$\hat{μ}(wA^ℕ) = μ(w)$ holds for each word $w ∈ A^*$.  In the rest of the
paper, we use the same symbol for $μ$ and~$\hat{μ}$.  A probability
measure~$μ$ is said to be \emph{(shift) invariant} if the equality
\begin{displaymath}
  ∑_{a ∈ A}{μ(aw)} = μ(w)
\end{displaymath}
holds for each word $w ∈ A^*$.  The support $\supp{μ}$ of a measure~$μ$
is the set $\supp{μ} = \{ w ∈ A^* : μ(w) > 0 \}$ of finite words.

The column vector such that each of its entry is~$1$ is denoted by~$𝟙$.  A
$P$-vector~$λ$ is called \emph{stochastic} (respectively,
\emph{substochastic}) if its entries are non-negative and sum up to~$1$
(respectively, to at most~$1$).  that is, $0 ⩽ λ_p ⩽ 1$ for each $p ∈ P$
and $λ𝟙 = 1$ (respectively, $λ𝟙 ⩽ 1$).  A matrix~$M$ is called
\emph{stochastic} (respectively, \emph{substochastic}) if each of its rows
is stochastic (respectively, \emph{substochastic}), that is $M𝟙 = 𝟙$
(respectively, $M𝟙 ⩽ 𝟙$). It is called \emph{strictly substochastic} if it
is substochastic but not stochastic.  This means that the entries of at
least one of its rows sum up to a value which is strictly smaller than~$1$.

In the paper, we mainly consider rational measures also known as hidden
Markov measures that we now introduce.  A measure~$μ$ is \emph{rational} if
it is realized by a weighted automaton \cite[Chap~.4]{Sakarovitch09a}.
Equivalently there is an integer~$m$, a row $(1 × m)$-vector~$π$, a
morphism~$ν$ from $A^*$ into $m × m$-matrices over real numbers and a
column $m × 1$-vector~$ρ$ such that the following equality holds for each
word $a_1 ⋯ a_k$ \cite{BerstelReutenauer10}.
\begin{displaymath}
  μ(a_1 ⋯ a_k) = π ν(a_1 ⋯ a_k) ρ  = π ν(a_1)  ⋯ ν(a_k) ρ
\end{displaymath}
The triple $\tuple{π,ν,ρ}$ is called a \emph{representation} of the
rational measure~$μ$.  By the the main result in \cite{HanselPerrin90}, it
can always be assumed that both the vector~$π$ and the matrix $∑_{a∈
  A}{ν(a)}$ are stochastic and that the vector~$ρ$ is the vector~$𝟙$.  The
triple $\tuple{π,ν,𝟙}$ is then called a \emph{stochastic representation}
of~$μ$.  The measure~$μ$ is invariant if $π ∑_{a∈ A}{ν(a)} = π$.

\subsection{Automata and ambiguity}

We refer the reader to \cite{PerrinPin04} for a complete introduction to
automata accepting (infinite) sequences of symbols.  A \emph{(Büchi)
  automaton} $𝒜$ is a tuple $\tuple{Q,A,Δ,I,F}$ where $Q$ is the finite
state set, $A$ the alphabet, $Δ ⊆ Q × A × Q$ the transition relation,
$I ⊆ Q$ the set of initial states and $F$ is the set of final states.  A
transition is a tuple $⟨ p,a,q ⟩$ in $Q × A × Q$ and it is written
$p \trans{a} q$.  A \emph{finite run} in~$𝒜$ is a finite sequence of
consecutive transitions,
\begin{displaymath}
   q_0 \trans{a_1} q_1 \trans{a_2} q_2 ⋯ q_{n-1} \trans{a_n} q_n
\end{displaymath}
Its \emph{label} is the word $a_1 a_2 ⋯ a_n$.  An \emph{infinite run}
in~$𝒜$ is a sequence of consecutive transitions,
\begin{displaymath}
  q_0 \trans{a_1} q_1 \trans{a_2} q_2 \trans{a_3} q_3 ⋯ 
\end{displaymath}
A run is \emph{initial} if its first state~$q_0$ is initial, that is,
belongs to~$I$.  A run is called \emph{final} if it visits infinitely often
a final state.  An infinite run is \emph{accepting} if it is both initial
and final.  A sequence is \emph{accepted} if it is the label of an
accepting run.  The set of accepted sequences is said to be \emph{accepted}
by the automaton.  As usual, an automaton is \emph{deterministic} if it has
only one initial state, that is $\#I = 1$ and if $p \trans{a} q$ and $p
\trans{a} q'$ are two of its transitions with the same starting state and
the same label, then $q = q'$.  The automaton pictured in
Figure~\ref{fig:not-closure} accepts the set $0^*1^ℕ$ of sequences having
some $0$s and then only~$1$s.  The leftmost automaton pictured in
Figure~\ref{fig:unambiguous} is deterministic while the middle one is not.
Both accept the set of sequences having infinitely many $1$s. An automaton
is \emph{trim} if each state occurs in an accepting run.

For each state~$q$, its \emph{future} (respectively \emph{bi-future}) is
the set $\future{q}$ (respectively, $\bifut{q}$) of sequences labelling a
final run (respectively, at least two final runs) starting from~$q$.  Let
$\clofut{q}$ (respectively, $\clobifut{q}$) be the set of sequences
labelling at least one (respectively, two) infinite run starting from~$q$
which might be final or not.  Note that if the automaton is trim,
$\clofut{q}$ is indeed the topological closure of~$\future{q}$ but that
$\clobifut{q}$ might not be the topological closure of~$\bifut{q}$ as shown
by the automaton pictured in Figure~\ref{fig:not-closure}.
\begin{figure}[htbp]
  \begin{center}
  \begin{tikzpicture}[>=stealth',initial text=,auto,inner sep=1pt]
    \tikzstyle{every state}=[semithick,minimum size=0.4]
    \node[state,initial left] (q0) at (0,0.4) {$0$};
    \node[state,accepting] (q1) at (1.5,0.8) {$1$};
    \node[state,accepting] (q2) at (1.5,0) {$2$};
    \path[->] (q0) edge[out=120,in=60,loop] node {$0$} (q0);
    \path[->] (q0) edge node[pos=0.6] {$1$} (q1);
    \path[->] (q0) edge node[pos=0.6,swap] {$1$} (q2);
    \path[->] (q1) edge[out=30,in=-30,loop] node {$1$} (q1);
    \path[->] (q2) edge[out=30,in=-30,loop] node {$1$} (q2);
  \end{tikzpicture}
  \end{center}
  \caption{$\clobifut{0} = 0^*1^ℕ$ and
           $\overline{\bifut{0}} = 0^*1^ℕ ∪ \{ 0^ℕ\}$}
  \label{fig:not-closure}
\end{figure}
The \emph{past} $\past{q}$ of a state~$q$ is the set of finite words
labelling a run ending in~$q$.  For an automaton~$𝒜$, we let $\fact{𝒜}$
denote the set of finite words labelling some run in~$𝒜$.  Therefore
$\fact{𝒜} = ⋃_{q ∈ Q}{\past{q}}$ where $Q$ is the state set of~$𝒜$.

An automaton is \emph{unambiguous} (for finite words) if for each states
$p,q ∈ Q$ and each word~$w$, there is at most one run $p \trans{w} q$
from~$p$ to~$q$ labelled by~$w$.  Each automaton which is either
deterministic or reverse-deterministic is unambiguous.

\begin{figure}[htbp]
  \begin{center}
  \begin{tikzpicture}[initial text=,auto,inner sep=1pt]
  \tikzstyle{every state}=[->,>=stealth',semithick,minimum size=0.4]
    \begin{scope}
      \node[state,initial above] (q11) at (0,0.75) {$1$};
      \node[state,accepting] (q12) at (1.5,0.75) {$2$};
      \path[->,>=stealth'] (q11) edge[out=210,in=150,loop] node {$0$} (q11);
      \path[->,>=stealth'] (q11) edge[bend left=15] node {$1$} (q12);
      \path[->,>=stealth'] (q12) edge[bend left=15] node {$0$} (q11);
      \path[->,>=stealth'] (q12) edge[out=30,in=-30,loop] node {$1$} (q12);
    \end{scope}
    \draw (2.7,-0.5) -- (2.7,2);
    \begin{scope}[xshift=110]
      \node[state,initial above] (q21) at (0,0.75) {$1$};
      \node[state,initial above,accepting] (q22) at (1.5,0.75) {$2$};
      \path[->,>=stealth'] (q21) edge[out=210,in=150,loop] node {$0$} (q21);
      \path[->,>=stealth'] (q21) edge[bend left=15] node {$0$} (q22);
      \path[->,>=stealth'] (q22) edge[bend left=15] node {$1$} (q21);
      \path[->,>=stealth'] (q22) edge[out=30,in=-30,loop] node {$1$} (q22);
    \end{scope}
    \draw (6.5,-0.5) -- (6.5,2);
    \begin{scope}[xshift=220]
      \node[state,initial above,accepting] (q31) at (0,1.5) {$1$};
      \node[state,accepting] (q32) at (0,0) {$2$};
      \node[state,accepting] (q33) at (1.7,1.5) {$3$};
      \node[state,accepting] (q34) at (1.7,0) {$4$};
      \path[->,>=stealth'] (q31) edge[out=210,in=150,loop] node {$0$} (q31);
      \path[->,>=stealth'] (q31) edge[bend left=15] node {$1$} (q32);
      \path[->,>=stealth'] (q32) edge[bend left=15] node {$0$} (q31);
      \path[->,>=stealth'] (q31) edge node {$0$} (q33);
      \path[->,>=stealth'] (q33) edge[bend left=15] node {$1$} (q34);
      \path[->,>=stealth'] (q34) edge[bend left=15] node {$0{,}1$} (q33);
      \path[->,>=stealth'] (q34) edge[bend left=8] node[pos=0.35] {$1$} (q31);
    \end{scope}
  \end{tikzpicture}
  \end{center}
  \caption{Three unambiguous automata}
  \label{fig:unambiguous}
\end{figure}
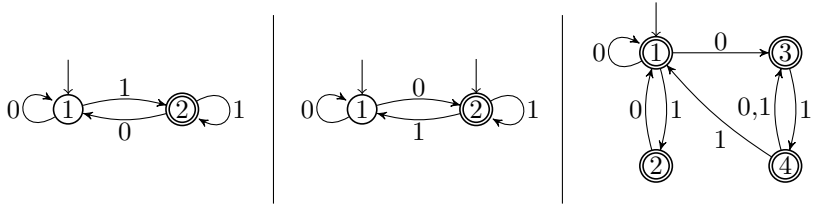

The three automata pictured in Figure~\ref{fig:unambiguous} are
unambiguous.  The leftmost one is deterministic and $\future{1} =
\future{2} = (0^*1)^ℕ$, $\clofut{1} = \clofut{2} = \{0,1\}^ℕ$ and
$\bifut{1} = \bifut{2} = ∅$.  The middle one is reverse deterministic
(that is, becomes deterministic if transitions are reversed), $\future{1} =
0(0^*1)^ℕ$, $\future{2} = 1(0^*1)^ℕ$, $\clofut{1} = 0\{0,1\}^ℕ$,
$\clofut{2} = 1\{0,1\}^ℕ$ and $\bifut{1} = \bifut{2} = ∅$.  The rightmost
one is neither deterministic nor reverse deterministic but it is
unambiguous.  Note however that $\bifut{1}$ is not empty: $\bifut{1} ⊃
0^*(01)^ℕ$.  An ambiguous automaton is pictured in
Figure~\ref{fig:example-ambiguous} below.

With each stochastic representation $\tuple{π,ν,𝟙}$ of a rational measure
is associated an automaton whose state set is $P = \{1,…,m\}$ where $m$ is
the common dimension of all matrices $ν(a)$.  For each states $p,q ∈ P$,
there is a transition $p \trans{a} q$ whenever $ν(a)_{p,q} > 0$.  The
initial states are those states~$q$ in~$P$ such that $π_q > 0$.  Due to
this automaton, $\clofut{p}$ is well-defined for a state~$p ∈ P$.  The
representation is called \emph{irreducible} if this automaton is strongly
connected.  A rational measure is called \emph{irreducible} if it has at
least one irreducible representation.

A strongly connected component~$C$ of graph (respectively automaton) is
called \emph{terminal} if it cannot be left, that is, if $p \trans{} q$
is a edge with $p ∈ C$, then $q ∈ C$.

\section{Main result} \label{sec:result}

Ambiguity of automata has been defined using finite words: an automaton is
ambiguous if some finite word~$w$ is the label of two different runs from a
state~$p$ to a state~$q$.  If the automaton is trim, this implies that some
sequence of the form $wy$ is the label of two different runs from~$p$.  The
converse of this implication does not hold in general.  The third automaton
pictured in Figure~\ref{fig:unambiguous} is unambiguous although the sequence
$(01)^ℕ = 0101⋯$ is the label of the following two accepting runs starting
from state~$1$.
\begin{align*}
  1 & \trans{0} 1 \trans{1} 2 \trans{0} 1 \trans{1} 2 \trans{0} 1 \trans{1}
      ⋯ \\
  1 & \trans{0} 3 \trans{1} 4 \trans{0} 3 \trans{1} 4 \trans{0} 3 \trans{1}
      ⋯ 
\end{align*}
However, the set $\bifut{1}$ is contained in $(0+1)^*(01)^ℕ$ and it is thus
countable and of measure~$0$ for the uniform measure.  Note that if each
transition $p \trans{1} q$ is replaced by the two transitions
$p \trans{1} q$ and $p \trans{2} q$, the set $\bifut{1}$ is not anymore
countable but it is still of measure~$0$ as a subset of $\{0, 1, 2\}^ℕ$.

The following theorem provides a characterization of ambiguity using
measure theory.  More precisely, it states that a strongly connected
automaton is unambiguous whenever the measure of sequences labelling two
runs is negligible, that is, of measure zero.
\begin{theorem} \label{thm:measure}
  Let $𝒜$ be a strongly connected Büchi automaton and let $μ$ be an
  irreducible rational measure such that $\supp{μ} = \fact{𝒜}$.  The
  following conditions are equivalent.
  \begin{itemize} \itemsep0cm
  \item[i)] The automaton~$𝒜$ is unambiguous.
  \item[ii)] For each state $q$ of~$𝒜$, $μ(\clobifut{q}) = 0$.
  \item[iii)] There is a state~$q$ of~$𝒜$ such that $μ(\bifut{q}) = 0$.
  \end{itemize}
\end{theorem}

The irreducibility of the measure ensures that it does not put too much
weight on too small sets (See example after
Proposition~\ref{pro:same-measure}).  The measure used to quantify this
ambiguity must also be compatible with the automaton.  More precisely, its
support must be equal to the set of finite words labelling at least one run
in the automaton.  If this condition is not fulfilled, the result may not
hold as it is shown by the following two examples.

Consider again the third automaton pictured in Figure~\ref{fig:unambiguous}.
Let $μ$ be the probability measure putting weight~$1/2$ on each of the
sequences $(01)^ℕ$ and~$(10)^ℕ$ and zero everywhere else.  More formally,
it is defined $μ((01)^ℕ) = μ((10)^ℕ) = 1/2$ and $μ(\{0, 1\}^ℕ ∖ \{(01)^ℕ,
(10)^ℕ\}) = 0$.  The measure $μ(\bifut{1}) = 1/2$ is non-zero although the
automaton is unambiguous because the support $(01)^* + (10)^*$ of this
measure~$μ$ is strictly contained in the set of words labelling a run in
this automaton.  This latter set is actually the set $\{0, 1\}^*$ of all
finite words over $\{0, 1\}$.

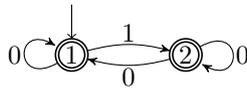
\begin{figure}[htbp]
  \begin{center}
  \begin{tikzpicture}[initial text=,auto,inner sep=1pt]
  \tikzstyle{every state}=[->,>=stealth',semithick,minimum size=0.4]
    \node[state,initial above,accepting] (q1) at (0,0.75) {$1$};
    \node[state,accepting] (q2) at (1.5,0.75) {$2$};
    \path[->,>=stealth'] (q1) edge[out=210,in=150,loop] node {$0$} (q1);
    \path[->,>=stealth'] (q1) edge[bend left=15] node {$1$} (q2);
    \path[->,>=stealth'] (q2) edge[bend left=15] node {$0$} (q1);
    \path[->,>=stealth'] (q2) edge[out=30,in=-30,loop] node {$0$} (q2);
  \end{tikzpicture}
  \end{center}
  \caption{An ambiguous automaton accepting $(0 + 10)^ℕ$}
  \label{fig:example-ambiguous}
\end{figure}

Consider the automaton pictured in Figure~\ref{fig:example-ambiguous}.  It
accepts the set~$X$ of sequences with no consecutive $1$s.  It is ambiguous
because the word~$00$ is the label of the two runs $2 \trans{0} 1 \trans{0}
1$ and $2 \trans{0} 2 \trans{0} 1$.  The uniform measure $μ(X)$ is
zero.  Therefore, both numbers $μ(\bifut{1})$ and~$μ(\bifut{2})$ are zero
although the automaton is ambiguous.  This comes from the fact that the
support $\{0,1\}^*$ of the uniform measure strictly contains the set
$\fact{X}$.  This latter set is the set $(0 + 10)^*(1+ \emptyword)$ of
finite words with no consecutive~$1$s.

\section{Reduction to closed sets} \label{sec:reduction}

The purpose of this section is to show that the measure $μ(X)$ of a
rational set of sequences is closely related to the measure
$μ(\overline{X})$ of its closure as long as the measure~$μ$ is compatible
with~$X$.  The main result of this section is the following proposition
which is used in the proof of Theorem~\ref{thm:measure}.  The rest of the
section is devoted to the proof of the proposition.

\begin{proposition} \label{pro:same-measure}
  Let $𝒜$ be a strongly connected Büchi automaton, $q$ be a state of~$𝒜$
  and $w$ be a word in~$\past{q}$. Let $μ$ be an irreducible rational
  measure such that $\supp{μ} = \fact{𝒜}$.  Then $μ(w\future{q}) =
  μ(w\clofut{q})$.
\end{proposition}

The following example shows that the irreducibility assumption of the
measure is indeed necessary.
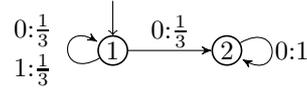
\begin{figure}[htbp]
  \begin{center}
  \begin{tikzpicture}[initial text=,auto,inner sep=1pt]
  \tikzstyle{every state}=[->,>=stealth',semithick,minimum size=0.4]
    \node[state,initial above] (q1) at (0,0.75) {$1$};
    \node[state] (q2) at (1.5,0.75) {$2$};
    \path[->,>=stealth'] (q1) edge[out=210,in=150,loop] node {$
      \begin{array}{c}
        0{:}\frac{1}{3} \\[1mm]
        1{:}\frac{1}{3}
      \end{array}$} (q1);
    \path[->,>=stealth'] (q1) edge node {$0{:}\frac{1}{3}$} (q2);
    \path[->,>=stealth'] (q2) edge[out=30,in=-30,loop] node {$0{:}1$} (q2);
  \end{tikzpicture}
  \end{center}
  \caption{A weighted automaton defining a non-irreducible measure}
  \label{fig:reducible-measure}
\end{figure}
Consider the measure given by the weighted automaton pictured in
Figure~\ref{fig:reducible-measure}.  This measure~$μ$ is equivalently
defined by $μ(w) = (1, 0)ν(w)\left(\begin{smallmatrix} 1 \\ 1
  \end{smallmatrix}\right)$ for each finite word~$w$ where the morphism~$ν$
from~$\{0,1\}^*$ into $2×2$-matrices is given by
\begin{displaymath}
  ν(0) = \left(\begin{smallmatrix} \frac{1}{3} & \frac{1}{3} \\ 0 & 1 \end{smallmatrix}\right)
  \quad\text{and}\quad
  ν(1) = \left(\begin{smallmatrix} \frac{1}{3} & 0 \\ 0 & 0 \end{smallmatrix}\right).
\end{displaymath}

The weight of a word~$w1$ ending with
a~$1$ is $3^{-|w|-1}$.  Therefore, the sum of these weights when $w$ ranges
over all words of length~$k$ over the alphabet~$\{0,1\}$ is given by
\begin{align*}
  ∑_{|w| = k}{μ(w1)} & = \frac{2^k}{3^{k+1}} \\
  ∑_{|w| ⩾ n}{μ(w1)} & = ∑_{k ⩾ n}{\frac{2^k}{3^{k+1}}} = \frac{2^n}{3^n} 
\end{align*}
Let $X = (0^*1)^ℕ$ be the set of sequences having infinitely many
occurrences of the symbol~$1$.  This set is equal to $\future{1}$ in the
leftmost automaton pictured in Figure~\ref{fig:unambiguous}.  Since $X$ is
contained in the union $⋃_{|w| ⩾ n}{w1\{0,1\}^ℕ}$ for each integer~$n ⩾ 0$,
the measure of~$X$ satisfies $μ(X) ⩽ 2^n/3^n$ for each integer~$n ⩾ 0$.
This proves that the measure of~$X$ is~$0$ although the measure of its
closure $\overline{X} = \{0,1\}^ℕ$ is~$1$.

If $𝒜$ is an automaton and $P ⊆ Q$ is a subset of its state set~$Q$, we let
$P ⋅ w$ denote the subset $P' ⊆ Q$ defined by $P' = \{ q : ∃ p ∈ P \;\; p
\trans{w} q \}$.  If $P$ is a singleton set~$\{q\}$, we write $q ⋅ w$ for
$\{q\} ⋅ w$.  By a slight abuse of notation, we also write $q ⋅ w = p$ for
$\{q\} ⋅ w = \{ p\}$. If $𝒜$ is deterministic, $q ⋅ w$ is either the empty
set or a singleton set.

\begin{lemma} \label{lem:alltrans}
  Let $𝒜$ be a strongly connected deterministic automaton.  There exists a
  finite word~$w$ such that:
  \begin{enumerate} \itemsep0cm
  \item[i)] there exists a state~$q$ such that $q⋅w$ is non-empty,
  \item[ii)] for each state~$q$ of~$𝒜$, if $q⋅w$ is non-empty, each
    transition of~$𝒜$ occurs in the run $q \trans{w} q⋅w$.
  \end{enumerate}
\end{lemma}
\begin{proof}
  Let $n$ be the number of states of~$𝒜$ and let $Q = \{ q_1, …, q_n \}$ be
  the state set of~$𝒜$.  We prove by induction on~$k$ that there exists a
  word~$w_k$ such that:
  \begin{enumerate} \itemsep0cm
  \item[i)] there exists a state~$q$ in $\{q_1, …, q_k\}$ such that $q⋅w$
    is non-empty,
  \item[ii)] for each state~$q$ in $\{q_1, …, q_k\}$, if $q⋅w$ is
    non-empty, each transition of~$𝒜$ occurs in the run $q \trans{w} q⋅w$.
  \end{enumerate}
  For $k = 1$, we choose a word~$w_1$ such that $q_1⋅w_1$ is non-empty and
  each transition of~$𝒜$ occurs in the run labelled by~$w_1$ from~$q_1$
  to~$q_1⋅w_1$.  Such a word does exist because $𝒜$ is strongly connected.
  Now suppose that $w_k$ has already been defined and let us define
  $w_{k+1}$ as follows.  If $q_{k+1} ⋅ w_k$ is empty, we set $w_{k+1} =
  w_k$.  If $q_{k+1} ⋅ w_k$ is non-empty, we set $w_{k+1} = w_ku_k$ where
  $u_k$ is a word such that $q_k ⋅ w_ku_k$ is non-empty and each transition
  of~$𝒜$ occurs in the run labelled by~$w_ku_k$ from~$q_k$ to~$q_k⋅w_ku_k$.
  In both cases, it is pure routine to check that the word~$w_{k+1}$
  satisfies the required property.
\end{proof}

Let $\tuple{π,ν,𝟙}$ be the representation of a rational measure~$μ$.  Its
support $\supp{ν}$ is defined by $\supp{ν} = \{ w : ν(w) ≠ 0 \}$.  It
obviously satisfies $\supp{μ} ⊆ \supp{ν}$.  This inclusion can be strict as
shown by the following example but it becomes an equality as soon as
$\supp{μ}$ is \emph{factorial}, that is closed under taking factor.

Let $μ$ be the measure defined by $μ(0w) = 0$ and $μ(1w) = 2^{-|w|}$ for
each word~$w$ in $\{0,1\}^*$.  It is rational because it is defined by
$μ(w) = (0,1)ν(w)\left(\begin{smallmatrix} 1 \\ 1 \end{smallmatrix}\right)$
where the morphism~$ν$ from~$\{0,1\}^*$ into $2×2$-matrices is given by
\begin{displaymath}
  ν(0) = \left(\begin{smallmatrix} & \frac{1}{2} \\ 0 & 0 \end{smallmatrix}\right)
  \quad\text{and}\quad
  ν(1) = \left(\begin{smallmatrix} 0 & 0 \\ \frac{1}{2} & \frac{1}{2} \end{smallmatrix}\right).
\end{displaymath}
The support of this measure~$μ$ is $\supp{μ} = 1\{0,1\}^*$ but the support
of the morphism~$ν$ is $\supp{ν} = \{ w ∈ \{0,1\}^* : ν(w) ≠ 0 \}$ is
$\{0,1\}^*$.  The support of a rational measure and the support of one of
its representations might not coincide in general but they do coincide as
soon as $\supp{μ}$ is factorial as stated by the following lemma.
\begin{lemma}
  Let $μ$ be an irreducible rational measure and let $\tuple{π,ν,𝟙}$ be an
  irreducible representation of~$μ$.  Then $\supp{μ}$ is factorial if
  and only if the equality $\supp{μ} = \supp{ν}$ holds.
\end{lemma}
\begin{proof}
  If the equality $\supp{μ} = \supp{ν}$ holds, then $\supp{μ}$ is
  obviously factorial because $\supp{ν}$ is factorial.
  
  We now prove the converse.  For each word~$w$, the equality $ν(w) = 0$
  implies that $μ(w) = 0$ and therefore the inclusion $\supp{μ} ⊆
  \supp{ν}$ always holds.  We now prove the reverse inclusion.  Let $w$
  be a word in~$\supp{ν}$.  There are two states $p$ and~$q$ such that
  $ν(w)_{p,q} > 0$.  Since $ν$ is irreducible, there is a state~$i$ and a
  word~$u$ such that $π_i > 0$ and $ν(u)_{i,p} > 0$.  These relations imply
  that $μ(uw) ⩾ π_iν(uw)_{i,q} ⩾ π_iν(u)_{i,p}ν(w)_{p,q} > 0$ and thus $w ∈
  \supp{μ}$ because $\supp{μ}$ is factorial.
\end{proof}
In the rest of the paper, the support of each measure is a factorial set
and both supports coincide.

The following lemma states the of sequences having finitely many occurrence
of some finite word has zero measure.  The irreducibility of the measure is
crucial.
\begin{lemma} \label{lem:zero-measure}
  Let $μ$ be an irreducible rational measure such that $\supp{μ}$ is
  factorial and let $w$ be a word in~$\supp{μ}$.  Then
  \begin{displaymath}
    μ(\{ x : |x|_w < ∞ \}) = 0.
  \end{displaymath}
  where $|x|_w$ is the number of occurrences of~$w$ in~$x$.
\end{lemma}
\begin{proof}
  Let $\tuple{π,ν,𝟙}$ be an irreducible representation of~$μ$ whose state
  set is~$P$.  For each $p ∈ P$, let $μ_p$ be the measure whose
  representation is $\tuple{δ_p,ν,𝟙}$ where $δ_p$ is defined by $δ_p(p') =
  1$ if $p = p'$ and $δ_p(p') = 0$ otherwise.  Note that the measure~$μ$ is
  equal to $μ = ∑_{p ∈ P}{π_pμ_p}$.
  
  We first prove that $μ(\{ x : |x|_w = 0 \}) = 0$. For that purpose, we
  introduce a deterministic Büchi automaton~$𝒜$ accepting the set $\{ x :
  |x|_w = 0 \}$.  Let $n$ be the length of the word~$w$.  Let $𝒜$ be the
  deterministic Büchi automaton $\tuple{Q,A,E,I,F}$ whose state set~$Q$ is
  the set $A^{⩽ n-1}$ of words of length less than~$n-1$.  The initial
  state is the empty word~$\emptyword$ and each state is final, that is, $F
  = Q$.  Its set~$E$ of transitions is defined by
  \begin{align*}
    E & {} =  \{ u \trans{a} ua : u ∈ A^{⩽n-2} ∧ a ∈ A \} \\
      & {} ∪ \{ bu \trans{a} ua : u ∈ A^{n-1} ∧ a,b ∈ A ∧ bua ≠ w \} 
  \end{align*}
  For each word~$u$, there is a run from the initial state~$\emptyword$
  labeled by~$u$ if $u$ contains no occurrence of~$w$.  The state reached
  by this run is either $u$ if $|u| ⩽ n-1$ or the suffix of length~$n-1$
  of~$u$.
  
  For each state $q$ in~$Q$, let $X_q$ be the set of sequences labelling a
  (accepting) run from~$q$ and $α_{p,q}$ be the measure $μ_p(X_q)$ for each
  $p ∈ P$ and $q ∈ Q$.  Note that $X_{\emptyword}$ is the set $\{ x : |x|_w
  = 0 \}$. For each $p ∈ P$ and $q ∈ Q$, the number $α_{p,q}$ satisfies the
  equality.
  \begin{displaymath}
    α_{p,q} = ∑_{a ∈ A, p' ∈ P} ν(a)_{p,p'}α_{p',q⋅a}
  \end{displaymath}
  All these equalities can be written using matrices.  The state set~$Q$
  of~$𝒜$ is split into $Q = Q_1 ⊎ Q_2$ where $Q_1 = A^{<n-1}$ and $Q_2 =
  A^{n-1}$.  Let $α$ be the vector $(α_{p,q})_{p ∈ P, q ∈ Q}$.  Let us
  write $α = (α_1,α_2)$ where $α_1$ and $α_2$ are the two vectors
  $(α_{p,q})_{p ∈ P, q ∈ Q_1}$ and $(α_{p,q})_{p ∈ P, q ∈ Q_2}$.
  \begin{displaymath}
    (α_1, α_2) =
    \begin{pmatrix}
      M_{1,1} & M_{1,2} \\
         0    & M_{2,2}
    \end{pmatrix}
    \begin{pmatrix}
      α_1 \\
      α_2
    \end{pmatrix}
  \end{displaymath}
  Since transitions such that $bua = w$ are not in~$𝒜$, the matrix
  $M_{2,2}$ is strictly sub-stochastic.  This implies that $α_{p,q} = 0$ for
  $q ∈ A^n$.  Since there are transitions from~$Q_1$ to~$Q_2$, the
  matrix~$M_{1,1}$ is also strictly sub-stochastic and the matrix
  $I-M_{1,1}$ is invertible.  It follows that $α_1$ satisfies $α_1 =
  (I-M_{1,1})^{-1}M_{2,2}α_2$ and that $α_1 = 0$.  This concludes the proof
  that that $μ(\{ x : |x|_w = 0 \}) = 0$.
  
  Let $k ⩾ 1$ be a positive integer.  The proof that $μ(\{ x : |x|_w ⩽ k
  \}) = 0$ is similar to the proof for $k = 0$.  We introduce a
  deterministic Büchi automaton~$𝒜_k$ accepting $\{ x : |x|_w ⩽ k \}$.  Its
  state set~$Q$ is $A^{<n-1} ∪ A^{n-1} × \{0, …, k\}$ and the initial state is
  the empty word~$\emptyword$.  Its set~$E_k$ of transitions is defined by
  \begin{align*}
    E & {} = \{ u \trans{a} ua : u ∈ A^{⩽ n-3} ∧ a ∈ A \} \\
      & {} ∪ \{ u \trans{a} (ua,0) : u ∈ A^{n-2} ∧ a ∈ A \} \\
      & {} ∪ \{ (bu,i) \trans{a} (ua,i) : u ∈ A^{n-2} ∧ a,b ∈ A ∧ bua ≠ w ∧
                0 ⩽ i ⩽ k\} \\
      & {} ∪ \{ (bu,i) \trans{a} (ua,i+1) : u ∈ A^{n-2} ∧ a,b ∈ A ∧ bua = w ∧
               0 ⩽ i ⩽ k-1\}
  \end{align*}
  For each word~$u$, there is a run from the initial state~$\emptyword$
  labeled by~$u$ if $u$ contains at most $k$ occurrences of~$w$. The state
  reached by this run is either $u$ if $|u| ⩽ n-1$ or the pair $(v,i)$
  where $v$ is the suffix of length~$n-1$ of~$u$ and $i = |u|_w$ is the
  number of occurrences of~$w$ in~$u$.  The state set~$Q$ of~$𝒜$ can be
  split into $Q = Q_1 ⊎ Q_2$ where $Q_1 = A^{<n-1} ∪ A^{n-1} × \{0, …,
  k-1\}$ and $Q_2 =A^{n-1} × \{k\}$.  The same reasonning can be used again
  to prove that $μ(\{ x : |x|_w ⩽ k \}) = 0$.
\end{proof}

We let $\trans{*}$ denote the accessibility relation in an automaton.  We
write $p \trans{*} q$ if there is a run from~$p$ to~$q$.  If $P$ and~$P'$
are two subsets of states of an automaton, we write $P \trans{*} P'$
whenever there is a run from a state in~$P$ to a state in~$P'$.  This
relation is not transitive in general but it is when each considered subset
is contained in a strongly connected component.  The following lemma gives
a property of Muller automata accepting the same set as a strongly
connected Büchi automaton.

\begin{lemma} \label{lem:maximal-component}
  Let $X ⊆ A^ℕ$ be a non-empty set of sequences accepted by a strongly
  connected Büchi automaton and let $w ∈ A^*$ be a finite word.  Let $𝒜$ be
  a Muller automaton accepting the set~$wX$ and let $𝒯$ be its table.  Let
  $F$ be an element of~$𝒯$ such that $F' ∈ 𝒯$ and $F \trans{*} F'$ imply
  $F' \trans{*} F$.  Then the strongly connected component containing $F$
  also belongs to the table~$𝒯$.
\end{lemma}
\begin{proof}
  Let $ℬ$ be a strongly connected Büchi automaton accepting~$X$.

  Let $w_1$ be the label of a run from the initial state of~$𝒜$ to a
  state~$q$ in~$F$.  Let $w_2$ be the label of a run from~$q$ to~$q$ such
  that the set of states visited by this run is exactly~$F$.  The sequence
  $w_1w_2^ℕ$ is thus accepted by~$𝒜$.  Since this set is contained
  in~$wA^ℕ$, $w$ must be a prefix of~$w_1w_2^n$ for $n$ large enough.  By
  replacing $w_1$ by $w_1w_2^n$, it can be assumed, without loss of
  generality, that $w$ is a prefix of~$w_1$.  Let $w'$ be the word such
  that $ww' = w_1$.
  
  Let $C$ be the strongly connected component containing~$F$.  We claim
  that for each word~$u$, $u$ is the label of some run in~$C$ if and only
  if $u$ is the label of some run in~$ℬ$.  Suppose first that $u$ is the
  label of some run in~$C$.  There are then two words $v_1$ and~$v_2$ such
  that $v_1uv_2$ is the label of a run from~$q$ to~$q$ in~$𝒜$.  Therefore
  the word $w_1v_1uv_2w_2^ℕ$ is accepted by~$𝒜$.  This proves that $u$ is
  the label of some run in~$ℬ$.  Now suppose that $u$ is the label of some
  run in~$ℬ$.  Since $w'w_2^ℕ$ is accepted by~$ℬ$, there is a run from the
  initial state of~$ℬ$ to some state~$q'$ labeled by~$w'$.  There are two
  words $v_1$ and~$v_2$ such that $v_1uv_2$ is the label of a run from~$q'$
  to~$q'$ visiting a final state of~$ℬ$.  Therefore the word
  $w'(v_1uv_2)^ℕ$ is accepted by~$ℬ$ and $w_1(v_1uv_2)^ℕ$ must be accepted
  by~$𝒜$.  This shows that $u$ is the label of some run in~$C$.
  
  By Lemma~\ref{lem:alltrans} applied to~$C$ considered as an automaton,
  there exists a word~$w_3$ such that, there exists a state~$p$ in~$C$ such
  that $p⋅w_3$ is well-defined and for each state~$p$ in~$C$, if $p⋅w_3$ is
  well-defined, then each transition in~$C$ occurs in the run from $p$ to
  $p⋅w_3$.  Applying the previous claim to the word~$w_3$ shows that $w_3$
  is the label of some run in~$ℬ$.  Therefore, there exist words $v_1$
  and~$v_2$ such that $w'(v_1w_3v_2)^ℕ$ is accepted by~$ℬ$ and
  $w_1(v_1w_3v_2)^ℕ$ is thus accepted by~$𝒜$.  Each occurrence of~$w_3$ in
  this sequence is the label of a run visiting all states of~$C$.  It
  follows that each state~$p$ of~$C$ occurs infinitely often in the run
  labelled by $w_1(v_1w_3v_2)^ℕ$ and thus $C$ belongs to the table~$𝒯$
  of~$𝒜$.
\end{proof}

A subset~$P$ of states of a Muller automaton~$𝒜$ is called \emph{essential}
if there is an infinite run in~$𝒜$ such that $P$ is the set of states that
occur infinitely often along this run.
\begin{proof}[Proof of Proposition~\ref{pro:same-measure}]
  The proof of the proposition is reduced to proving that $μ(w\clofut{q} ∖
  w\future{q}) = 0$.  We consider a Muller automaton accepting
  $w\future{q}$ with a table~$𝒯$.  Note that a Muller automaton accepting
  $w\clofut{q}$ is obtained by replacing the table~$𝒯$ by the table
  $\overline{𝒯}$ which contains each essential set of states which can
  access an essential set of states in~$𝒯$.  The difference set
  $w\clofut{q} ∖ w\future{q}$ is thus accepted by the same Muller automaton
  with the table $𝒯 ∖ \overline{𝒯}$.  By Lemma~\ref{lem:maximal-component},
  each maximal essential state is in the table~$𝒯$.  By combining Lemmas
  \ref{lem:alltrans} and~\ref{lem:zero-measure}, it is clear that
  $μ(w\clofut{q} ∖ w\future{q}) = 0$.
\end{proof}

\section{Proof for closed sets} \label{sec:proof}

Thanks to Proposition~\ref{pro:same-measure}, it is sufficient to study
closed sets.  As the initial states of the automaton are not relevant for
the statement of Theorem~\ref{thm:measure}, we consider automata where all
states are initial and final, that is, $I = F = Q$.  It turns out that
these automata accept shift spaces that we now introduce.

The \emph{shift map} is the function~$σ$ which maps each sequence $(x_i)_{i
  ⩾ 1}$ to the sequence $(x_{i+1})_{i ⩾ 1}$ obtained by removing its first
element. A \emph{shift space} is a subset $X$ of~$A^ℕ$ which is closed for
the usual product topology and such that $σ(X) = X$.  A classical example
of a shift space is the \emph{golden mean shift}: it is the set
$\{0,10\}^ℕ$ of sequences with no consecutive $1$s. We refer the reader to
\cite{LindMarcus95} for a complete introduction to shift spaces.

If a shift space is accepted by some trim Büchi automaton, it is also
accepted by the same automaton in which each state is made initial and
final.  A shift space is called \emph{sofic} if it is accepted by some
automaton.  A sofic shift is called \emph{irreducible} if it is accepted by
a strongly connected Büchi automaton.  It is well-known \cite{LindMarcus95}
that each shift space is characterized by the set of factors of its
sequences.  Let us recall that $\fact{X}$ denotes the set of factors of a
shift space~$X$.

There is a unique, up to isomorphism, deterministic automaton accepting an
irreducible sofic shift with the minimal number of states
\cite[Thm~3.3.18]{LindMarcus95}.  This \emph{minimal automaton} is also
referred to as either its Shannon cover or its Fischer cover.  It can be
obtained from any automaton accepting the shift space via determinizing and
state-minimizing algorithms, e.g., \cite[pp. 92]{LindMarcus95}, \cite[pp.
68]{HopcroftUllman79}.  The minimal automaton of the golden mean shift is
the leftmost automaton pictured in Figure~\ref{fig:goldenmean}.  A
\emph{synchronizing word} of a strongly connected automaton is a word~$w$
such that there is a unique state~$q$ such that $w ∈ \past{q}$.  The
word~$1$ is a synchronizing word of both automata pictured in
Figure~\ref{fig:goldenmean}.  The mini\-mal automaton of a sofic shift has
always at least one synchronizing word~\cite[Prop.~3.3.16]{LindMarcus95}.

\begin{figure}[htbp]
  \begin{center}
  \begin{tikzpicture}[initial text=,auto,inner sep=1pt]
  \tikzstyle{every state}=[->,>=stealth',semithick,minimum size=0.4]
    \begin{scope}
      \node[state,initial above,accepting] (q11) at (0,0.75) {$1$};
      \node[state,initial above,accepting] (q12) at (1.5,0.75) {$2$};
      \path[->,>=stealth'] (q11) edge[out=210,in=150,loop] node {$0$} (q11);
      \path[->,>=stealth'] (q11) edge[bend left=15] node {$1$} (q12);
      \path[->,>=stealth'] (q12) edge[bend left=15] node {$0$} (q11);
    \end{scope}
    \draw (3.1,0) -- (3.1,1.6);
    \begin{scope}[xshift=150]
      \node[state,initial above,accepting] (q21) at (0,0.75) {$1$};
      \node[state,initial above,accepting] (q22) at (1.5,0.75) {$2$};
      \path[->,>=stealth'] (q21) edge[out=210,in=150,loop] node {$0$} (q21);
      \path[->,>=stealth'] (q21) edge[bend left=15] node {$0$} (q22);
      \path[->,>=stealth'] (q22) edge[bend left=15] node {$1$} (q21);
    \end{scope}
  \end{tikzpicture}
  \end{center}
  \caption{Two automata accepting the golden mean shift}
  \label{fig:goldenmean}
\end{figure}
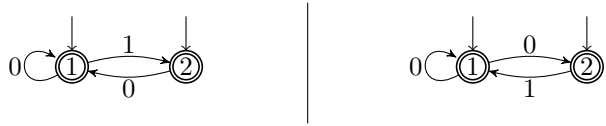

The next lemma states that some sets induced by the minimal automaton of a
sofic shift have a positive measure.  
\begin{lemma} \label{lem:minimal-positive}
  Let $r$ be a state of the minimal automaton of an irreducible sofic shift
  space~$X$ and let $w$ be a synchronizing word such that $w ∈ \past{r}$.
  Let $μ$ be a rational measure such that $\supp{μ} = \fact{X}$.  Then
  $μ(w\clofut{r}) = μ(wA^ℕ) > 0$.
\end{lemma}
\begin{proof}
  Since $X$ is a closed set, its complement $A^ℕ ∖ X$ is equal to the union
  $⋃_{w ∉ \fact{X}} wA^ℕ$.  Since $\supp{μ} = \fact{X}$, the equality
  $μ(X) = 1$ holds.  It follows that $μ(X ∩ wA^ℕ) = μ(wA^ℕ)$ for each word
  $w$.  We claim that if $w$ is synchronizing and $w ∈ \past{r}$, then
  $w\clofut{r} = X ∩ wA^ℕ $.  The inclusion $w\clofut{r} ⊆ X ∩ wA^ℕ $
  follows directly from $w ∈ \past{r}$.  The reverse inclusion follows from
  the fact that $w$ is synchronizing.  Combining the two equalities, we get
  that $μ(w\clofut{r}) = μ(wA^ℕ)$.  This latter number is positive since
  $w ∈ \fact{X}$ and $\supp{μ} = \fact{X}$.
\end{proof}

The following lemma establishes a link between any automaton accepting a
sofic system and its minimal automaton.  It allows us to tranfert the
result of the previous lemma to non-minimal automata.
\begin{lemma} \label{lem:witness}
  Let $𝒜$ be a strongly connected automaton accepting an irreducible sofic
  shift space~$X$.  Let $w$ be a word and $q$ be a state of~$𝒜$ such that
  $w ∈ \past{q}$.  There exists a state~$r$ of the minimal automaton of~$X$
  and a synchronizing word~$v$ of this minimal automaton such that $wv ∈
  \past{r}$ and $\clofut{q} ∩ vA^ℕ = v\clofut{r}$.
\end{lemma}
\begin{proof}
  Let $Q$ be the state set of~$𝒜$.  Let us consider the deterministic
  automaton $\hat{𝒜}$ whose state set~$\hat{Q}$ is the set of non-empty
  subsets of~$Q$ of the form $q ⋅ u$ for some word~$u ∈ A^*$. The
  transitions of~$\hat{𝒜}$ are the transitions of the form $q ⋅ u \trans{a}
  q ⋅ ua$ for each word~$u \in A^*$ and each symbol~$a ∈ A$.
  
  Let $C$ be a terminal strongly connected component of~$\hat{𝒜}$.  We
  claim that $C$, seen as a whole automaton, accepts~$X$.  Let $u_0$ be a
  word such that $q ⋅ u_0$ is a state of~$\hat{𝒜}$ in~$C$.  Let $u_2 ∈
  \fact{X}$ be a factor of~$X$.  Since $𝒜$ is strongly connected, there is
  a word~$u_1$ such that there is a run in~$𝒜$ starting from~$q$ and
  labelled by $u_0u_1u_2$.  This shows that $u_2$ is the label of a run
  from $q ⋅ u_0u_1$ to $q ⋅ u_0u_1u_2$ in~$C$ and that any factor
  of~$X$ is the label of some run in~$C$.  Conversely, each finite word
  labelling a run in~$C$ is also the label of a run in~$𝒜$.  This proves
  that $C$ accepts $X$.
  
  Let $∼$ be the equivalence relation on states of~$\hat{𝒜}$ defined by $P
  ∼ P'$ iff $\clofut{P} = \clofut{P'}$.  The automaton $C/\!\!∼$ is the
  minimal deterministic automaton of~$X$.  As $C/\!\!∼$ is a minimal
  automaton there is a synchronizing word~$u'_1$ such that $q ⋅ u_0u'_1$ is
  also a state in~$C$.
  
  Let $v$ be the word $u_0u'_1$.  Since $u'_1$ is synchronizing, $v$ is
  also synchronizing. The equality $\clofut{q} ∩ vA^ℕ = v\clofut{r'}$ holds
  where $r' = q ⋅ v$.  Since $wv$ is a factor of~$X$, there is some
  state~$r$ of~$C$ such that $wv ∈ \past{r}$.  Since $u'_1$ is
  synchronizing in~$C/\!\!∼$, the states $r$ and~$r'$ satisfy $r ∼ r'$ and
  thus $\clofut{r} = \clofut{r'}$.  It follows that $wv ∈ \past{r}$ and
  $\clofut{q} ∩ vA^ℕ = v\clofut{r}$
\end{proof}

Combining Lemmas \ref{lem:minimal-positive} and~\ref{lem:witness} yields
the following result.
\begin{lemma} \label{lem:automaton-positive}
  Let $𝒜$ be an strongly connected automaton accepting an irreducible sofic
  shift~$X$. Let $q$ be a state of~$𝒜$ and let $w$ be a word such that $w ∈
  \past{q}$.  Let $μ$ be an irreducible rational measure such that $\supp{μ}
  = \fact{X}$.  Then $μ(w\clofut{q}) > 0$.
\end{lemma}
\begin{proof}
  By Lemma~\ref{lem:witness}, there exists a state~$r$ of the minimal
  automaton of~$X$ and a synchronizing word~$v$ such that $wv ∈ \past{r}$
  and $\clofut{q} ∩ vA^ℕ = v\clofut{r}$.  This implies that
  $w\clofut{q} ∩ wvA^ℕ = wv\clofut{r}$.  By
  Lemma~\ref{lem:minimal-positive}, the measure $μ(wv\clofut{r})$ is
  positive and thus $μ(w\clofut{q}) > 0$.
\end{proof}

It is a very classical result that not all regular sets of sequences are
accepted by deterministic Büchi automata. This is the reason why Muller
automata with a more involved acceptance condition were introduced.
Landweber's theorem states that a regular set of sequences is accepted by a
deterministic Büchi automaton if and only it is a $G_δ$-set (that is
$Π_0^2$) \cite[Thm~I.9.9]{PerrinPin04}.  This implies in particular that
regular and closed sets\footnote{Not to be confused with \emph{regular
    closed sets} which are equal to the closure of their interior
  \cite[Chap.~4]{Halmos63}.} are accepted by deterministic Büchi automata.
Regular and closed sets are actually accepted by deterministic Büchi
automata in which each state is final \cite[Prop~III.3.7]{PerrinPin04}.

Lemma~\ref{lem:automaton-positive} states that the future of a state in
automaton with all states final has a positive measure.  The following
provides a converse.  It states that a closed set~$F$ with positive measure
contains the future of a state of the minimal automaton, prefixed by some
word~$w$.  The word~$w$ is really needed because the closed set~$F$ can be
arbitrary small.
\begin{lemma} \label{lem:positive-measure}
  Let $X$ be a sofic shift space and let $μ$ be an irreducible rational
  measure such that $\supp{μ} = \fact{X}$.  Let $F$ be a regular and closed
  set contained in~$X$.  If $μ(F) > 0$, there exists a word~$w$ and a
  state~$r$ of the minimal automaton of~$X$ such that $w ∈ \past{r}$ and
  $w\clofut{r} ⊆ F$.
\end{lemma}

Before proceeding to the proof of the lemma, we show that even in the case
of the full shift, that is $X = A^ℕ$, both hypothesis of being regular and
closed are necessary.  Since the minimal automaton of the full shift has a
single state~$r$ satisfying $\clofut{r} = A^ℕ$, the lemma can be, in that
case, rephrased as follows.  If $μ(F) > 0$ where $μ$ is the uniform
measure, then there exists a word~$w$ such that $w A^ℕ⊆ F$.

Being regular is of course not sufficient because the set $(0^*1)^ℕ$ of
sequences having infinitely many occurrences of~$1$ is regular and has
measure~$1$ but does not contain any set of the form $wA^ℕ$.  Being closed
is also not sufficient as it is shown by the following example.  Let $X$ be
the set of sequences such that none of their non-empty prefixes of even
length is a palindrome.  The complement of~$X$ is equal to the following
union
\begin{displaymath}
  ⋃_{n ⩾ 1} Z_n
  \quad\text{where}\quad
  Z_n = ⋃_{|w| = n} w\tilde{w}A^ℕ
\end{displaymath}
and where $\tilde{w}$ stands for the reverse of~$w$.  Suppose for instance
that the alphabet is $A = \{0,1\}$.  The measure of~$Z_n$ is equal $2^{-n}$
because there are $2^n$ words of length~$n$ and the measure of each
cylinder $w\tilde{w}A^ℕ$ is $2^{-2n}$.  Furthermore, the set $Z_1 ∪ Z_2$ is
equal to $00A^ℕ ∪ 11A^ℕ ∪ 0110A^ℕ ∪ 1001A^ℕ$ whose measure is $5/8$.  This
shows that the measure of the complement of~$X$ is bounded by
$5/8 + ∑_{n ⩾ 3} 2^{-n} = 7/8$ (Note that this is really an upper bound as
the sets~$Z_n$ are not pairwise disjoint). Therefore $X$ has a positive
measure but it does not contain any cylinder.  Indeed, in each
cylinder~$wA^ℕ$, the cylinder $w\tilde{w}A^ℕ$ is out of~$X$.

\begin{proof}
  Let $\tuple{π,ν,𝟙}$ be a stochastic representation of dimension~$m$ of
  the rational measure~$μ$.  Let $P$ be the set $\{1,…,m\}$.  For each
  $p ∈ P$, we let $μ_p$ be the measure whose representation is
  $\tuple{δ_p,ν,𝟙}$ where the row vector $δ_p$ is given by $(δ_p)_{p'} = 1$
  if $p' = p$ and $0$ otherwise.  The measure $μ$ satisfies the equality
  $μ = ∑_{p = 1}^m π_p μ_p$.
  
  Let $𝒜$ be a deterministic Büchi automaton accepting~$F$ whose state set
  is $Q$.  The unique initial state of~$𝒜$ is~$i$.  Since $F$ is closed, it
  can be assumed that all states of~$𝒜$ are final.  For each state~$q$
  of~$𝒜$, the set $\clofut{q}$ is the set of accepted sequences if $q$ is
  taken as the unique initial state of the automaton.

  We consider a weighted graph~$𝒢$ whose vertex set is $P × Q$.
  The weight of the edge from the vertex $(p,q)$ to the the vertex
  $(p',q')$ is given by
  \begin{displaymath}
    \mathsf{w}_{p,q,p',q'} = ∑_{q \trans{a} q'} ν(a)_{p,p'}
  \end{displaymath}
  where the summation ranges over all transitions $q \trans{a} q'$ in the
  automaton~$𝒜$.  Since $μ(F) > 0$, there exists at least one integer~$p$
  such that $μ_p(F) > 0$.  Without loss of generality, it can be assumed
  that this integer~$p$ is $p = 1$.  The vertex $(1,i)$ where $i$ is the
  initial state~$𝒜$ is called the \emph{initial vertex} of~$𝒢$.  The
  graph~$𝒢$ is restricted to its accessible part from its initial vertex
  $(1,i)$, that is, the set of vertices $(p,q)$ such that there is a path
  from~$(1,i)$ to~$(p,q)$ made of edges with positive weight.  Vertices
  which are not accessible from~$(1,i)$ are ignored in the rest of this
  proof.

  With each vertex $(p,q)$ of~$𝒢$ is associated the real number
  $α_{p,q} = μ_p(\clofut{q})$.  Let $α$ be the row $P × Q$-vector whose
  entries are the numbers~$α_{p,q}$.  Let $M$ be the matrix of weights of
  the graph~$𝒢$: the $((p,q),(p',q'))$-entry of~$M$ is the weight
  $\mathsf{w}_{p,q,p',q'}$ defined above.  The matrix~$M$ and the
  vector~$α$ satisfy the equality $α = Mα$.  This latter equality comes
  first from the equality
  \begin{displaymath}
    \clofut{q} = ⨄_{q \trans{a} q'} a\clofut{q'}
  \end{displaymath}
  for each
  state~$q$ of~$𝒜$ where $⊎$ stands for disjoint union and second
  from the equality
  \begin{displaymath}
    μ_p(aF) = ∑_{p' = 1}^m ν(a)_{p,p'}μ_{p'}(F)
  \end{displaymath}
  for each $p ∈ P$, each symbol~$a$ and each measurable set~$F$.

  We claim that if $(p,q)$ belongs to a terminal strongly connected
  component~$C$ of~$𝒢$ and $α_{p,q} > 0$, then $\clofut{p} ⊆ \clofut{q}$.
  Let $α'$ and $M'$ be the restrictions of~$α$ and~$M$ respectively to the
  vertices in~$C$.  Because $C$ is terminal, the equality $α' = M'α'$
  holds.  If the matrix~$M'$ is strictly substochastic, this latter equality
  implies that $α$ is the zero vector and this would contradict
  $α_{p,q} > 0$.  The matrix~$M'$ is then stochastic.  The sum of the
  elements of the $(p,q)$-row of the matrix~$M'$ is equal to
  \begin{displaymath}
    ∑_{p',q'} \mathsf{w}_{p,q,p',q'} = ∑_{q \trans{a} q'} ∑_{p' = 1}^m ν(a)_{p,p'}.
  \end{displaymath}
  Since the automaton~$𝒜$ is deterministic the subset $q ⋅ a$ is either the
  empty set or a singleton set $\{ q' \}$.  This means that $q$ and~$a$
  being fixed, there is at most one choice for~$q'$.  Let us denote by
  $β_{p,a}$ the sum $∑_{p' = 1}^m ν(a)_{p,p'}$ so that
  \begin{displaymath}
    ∑_{p',q'} \mathsf{w}_{p,q,p',q'} = ∑_{q \trans{a} q'}{β_{p,a}}.
  \end{displaymath}
  Since the matrix $∑_{a ∈ A}ν(a)$ is stochastic, the sum
  $∑_{a ∈ A}{β_{p,a}}$ is equal to~$1$.  The sum
  $∑_{p',q'} \mathsf{w}_{p,q,p',q'}$ is thus equal to~$1$ if for each
  symbol~$a$, $β_{p,a} > 0$ implies that $q ⋅ a$ is not empty.  We claim
  that if $M'$ is stochastic, then $\clofut{p} ⊆ \clofut{q}$ for each
  vertex $(p,q)$ in~$C$.  Let $x = a_1a_2a_3⋯$ be sequence in~$\clofut{p}$.
  Then there exists a sequence $p = p_0,p_1,p_2,…$ in $P^ℕ$ such that
  $ν(a_i)_{p_{i-1},p_i} > 0$ for each $i ⩾ 1$.  This last relation implies
  that $β_{p_{i-1},a_i} > 0$.  There exists then a (unique) sequence
  $q = q_0,q_1,q_2,…$ of states of~$𝒜$ such that $q_{i+1} = q_i ⋅ a_{i+1}$.
  This completes the proof of the claim.

  Now we complete the proof.  Let $V_1$ (respectively, $V_2$) be the set of
  vertices in a non-terminal (respectively, terminal) strongly connected
  component of~$𝒢$.  We write $α = (\bar{α}_1,\bar{α}_2)$ where the vectors
  $\bar{α}_1$ and~$\bar{α}_2$ are respectively
  $\bar{α}_1 = (α_v)_{v ∈ V_1}$ and $\bar{α}_2 = (α_v)_{v ∈ V_2}$.  The
  relation $α = Mα$ is equivalent to the relations
  \begin{align*}
    \bar{α}_1 & = M_1\bar{α}_1 + M_3\bar{α}_2 \\
    \bar{α}_2 & = M_2\bar{α}_2
  \end{align*}
  where $M_1$ (respectively, $M_2$) is the restriction of~$M$ to rows and
  columns indexed by~$V_1$ (respectively, $C_2$) and $M_3$ is the
  restriction of~$M$ to rows indexed by~$V_1$ and columns indexed by~$V_2$.
  Since there is at least one transition from~$Q_1$ to~$Q_2$, the
  matrix~$M_1$ is strictly substochastic and its spectral radius is
  strongly less than~$1$.  The matrix $I - M_1$ is thus invertible. The
  first relation is thus equivalent to
  \begin{displaymath}
    \bar{α}_1 = (I - M_1)^{-1}M_3\bar{α}_2.
  \end{displaymath}
  This last equality shows that $\bar{α}_2 = 0$ implies $\bar{α}_1 = 0$ and
  thus $α = 0$.  Let $(p,q)$ be a vertex in~$V_2$ such that $α_{p,q} > 0$
  and thus $\clofut{p} ⊆ \clofut{q}$.
  
  There is then a path from $(1,i)$ to the state $(p,q)$.  There exists a
  word~$u$ such that $ν(u)_{1,p} > 0$ and $i \trans{u} q$ in~$𝒜$.  By
  Lemma~\ref{lem:witness}, there is a word~$v$ and a state~$r$ of the
  minimal automaton of~$X$ such that $uv ∈ \past{r}$ and $\clofut{p} ∩ vA^ℕ
  = v\clofut{r}$.  Note that applying Lemma~\ref{lem:witness} requires that
  the rational measure is iredducible.  Thus $v\clofut{r} ⊆ \clofut{q}$ and
  $uv\clofut{r} ⊆ u\clofut{q} ⊆ F$.  Setting $w = uv$ gives the result.
\end{proof}

The following result is trivially true when the measure~$μ$ is shift
invariant because $μ(F) = ∑_{|w|=m}μ(wF)$ but it does not hold in general.
\begin{lemma} \label{lem:add-prefix}
  Let $X$ be a sofic shift space and let $μ$ be an irreducible rational
  measure such that $\supp{μ} = \fact{X}$.  Let $F$ be a regular and closed
  set contained in~$X$.  If $μ(F) = 0$, then $μ(wF) = 0$ for each finite
  word~$w$.
\end{lemma}
\begin{proof}
  We prove that $μ(wF) > 0$ implies $μ(F) > 0$.  Suppose that $μ(wF) > 0$.
  Since $wF$ is also regular and closed, there exists, by
  Lemma~\ref{lem:positive-measure}, a word~$u$ and a state~$r$ of the
  minimal automaton of~$X$ such that $u ∈ \past{r}$ and $u\clofut{r} ⊆ wF$.
  This latter inclusion implies that either $u$ is a prefix of~$w$ or $w$
  is a prefix of~$u$.  In the first case, that is $w = uv$ for some
  word~$v$, the inclusion is equivalent to $\clofut{r} ⊆ vF$.  Let $s$ be
  state such that $r \trans{v} s$.  Then $v\clofut{s}$ is contained in
  $\clofut{r}$ and thus $\clofut{s} ⊆ F$.  By
  Lemma~\ref{lem:automaton-positive}, $μ(\clofut{s}) > 0$ and thus
  $μ(F) > 0$.  In the second case, that is $u = wv$, for some~$v$, the
  inclusion is equivalent to $v\clofut{r} ⊆ F$.  Again by
  Lemma~\ref{lem:automaton-positive}, $μ(v\clofut{r}) > 0$ and thus
  $μ(F) > 0$.
\end{proof}

The following lemma is an extension to pairs of futures of states of the
result of Lemma~\ref{lem:minimal-positive} for one state.
\begin{lemma} \label{lem:same-cylinder}
  Let $𝒜$ be strongly connected automaton accepting a shift space~$X$.  Let
  $μ$ be an irreducible rational measure such that $\supp{μ} = \fact{X}$.
  Let $q$ and $q'$ two states of~$𝒜$ such that $μ(\clofut{q} ∩ \clofut{q'})
  > 0$.  Then there exists a word~$w$ such that $\clofut{q} ∩ wA^ℕ =
  \clofut{q'} ∩ wA^ℕ$.
\end{lemma}
\begin{proof}
  We claim that there exists a word~$w$ and a state~$r$ of the minimal
  automaton of~$X$ such that $w ∈ \past{r}$ and
  \begin{displaymath}
    \clofut{q} ∩ wA^ℕ =  \clofut{q'} ∩ wA^ℕ = w\clofut{r}.
  \end{displaymath}
  
  Let $F$ be the closed set $\clofut{q} ∩ \clofut{q'}$.  By
  Lemma~\ref{lem:positive-measure} applied to~$F$, there exists a word~$u$
  and a state~$s$ of the minimal automaton of~$X$ such that $u ∈ \past{s}$
  and
  \begin{align*}
    u\clofut{s} & ⊆ \clofut{q} \\
    u\clofut{s} & ⊆ \clofut{q'} 
  \end{align*}
  Let $v$ be a synchronizing word of the minimal automaton of~$X$ such that
  $s ⋅ v$ is not empty.  Let $w$ be the word $uv$ and let $r$ be the state
  $s ⋅ v$.  Since $u ∈ \past{s}$ and $r = s ⋅ v$, $w ∈ \past{r}$.  We claim
  that $\clofut{q} ∩ wA^ℕ = w\clofut{r}$.  Suppose first that $x$ belongs
  to $\clofut{q} ∩ wA^ℕ$.  The sequence~$x$ is then equal to $wx'$ for some
  sequence~$x'$ and it is the label of a run in the minimal automaton
  of~$X$.  Since $w = uv$ and $v$ is synchronizing, the sequence~$x'$ must
  belong to $\clofut{r}$.  Suppose conversely that $x$ belongs
  to~$w\clofut{r}$.  It is then equal to $uvx'$ for some $x'$ in
  $\clofut{r}$.  Since $r = s ⋅ v$, $vx' ∈ \clofut{s}$. It follows from the
  inclusion $u\clofut{s} ⊆ \clofut{q}$ that $x$ belongs to $\clofut{q}$.
  This completes the proof of the equality
  $\clofut{q} ∩ wA^ℕ = w\clofut{r}$.  By symmetry, the equality
  $\clofut{q'} ∩ wA^ℕ = w\clofut{r}$ also holds and the proof is completed.
\end{proof}

This last lemma establishes a link between unambiguity of an automaton and
measures of futures of its states.  More precisely, it states that the
future of two states that can be reached from the same state and reading
the same word have an intersection of zero measure.  Its proof is more
combinatorial than previous ones.
\begin{lemma} \label{lem:ambiguity2}
  Let $𝒜$ be an unambiguous strongly connected automaton accepting a shift
  space~$X$.  Let $μ$ be an irreducible rational measure such that
  $\supp{μ} = \fact{X}$.  If there are two runs $p \trans{u} q$ and $p
  \trans{u} q'$, with $q ≠ q'$, then $μ(\clofut{q} ∩ \clofut{q'}) = 0$.
\end{lemma}
\begin{proof}
  Suppose by contradiction that $μ(\clofut{q} ∩ \clofut{q'}) > 0$.  There
  exists, by Lemma~\ref{lem:same-cylinder}, a word~$v$ such that
  $\clofut{q} ∩ vA^ℕ = \clofut{q'} ∩ vA^ℕ$.  Let $q ⋅ v$ (respectively
  $q' ⋅ v$) be the set $\{ q_1,…,q_r \}$ (respectively
  $\{q'_1,…,q'_{r'}\}$). Since $\clofut{q} ∩ vA^ℕ = \clofut{q'} ∩ vA^ℕ$,
  the equality
  $\clofut{q_1} ∪ ⋯ ∪ \clofut{q_r} = \clofut{q'_1} ∪ ⋯ ∪ \clofut{q'_{r'}}$
  holds. Since the automaton is strongly connected, there is a run
  $q_1 \trans{w} p$ from~$q_1$ to~$p$.  Combining this run with the run
  $p \trans{u} q \trans{v} q_1$ yields the cyclic run
  $q_1 \trans{wuv} q_1$.  Since
  $\clofut{q_1} ∪ ⋯ ∪ \clofut{q_r} = \clofut{q'_1} ∪ ⋯ ∪ \clofut{q'_{r'}}
  $, the sequence $(wuv)^ℕ = wuvwuv⋯$ belongs to a set~$\clofut{q'_i}$ for
  some $1 ⩽ i ⩽ r'$.  By symmetry, it can be assumed that
  $(wuv)^ℕ ∈ \clofut{q'_1}$.  There exists then a run starting from~$q'_1$
  with label $(wuv)^ℕ$.  This run can be decomposed
  \begin{displaymath}
    q'_1 \trans{wuv} p_1 \trans{wuv} p_2 \trans{wuv} p_3 ⋯.
  \end{displaymath}
  Since there are finitely many states, there are two integers $k,ℓ ⩾ 1$
  such that $p_k = p_{k+ℓ}$.  There are then the following two runs
  from~$p$ to $p_k = p_{k+ℓ}$ with the same label $(uvw)^{k+ℓ}uv$.
  \begin{align*}
    p & \trans{u} q \trans{v} q_1 \trans{(wuv)^{k-1}} q_1
        \trans{w} p \trans{uv} q'_1 \trans{(wuv)^ℓ} p_k \\
    p & \trans{u} q' \trans{v} q'_1 \trans{(wuv)^{k+ℓ}} p_{k+ℓ} 
  \end{align*}
  This is a contradiction with the fact that $𝒜$ is unambiguous.
\end{proof}

\begin{proof}[Proof of Theorem~\ref{thm:measure}]
  
  We first prove that (i) implies (ii).  suppose that the automaton~$𝒜$ is
  unambiguous.  We show that $μ(\clobifut{p}) = 0$ for each state~$p$.  We
  start by a decomposition of the set~$\clobifut{p}$.  Let $x = a_1a_2a_3⋯$
  be a sequence in~$\clobifut{p}$ and let $ρ$ and $ρ'$ be the two different
  runs labelled by~$x$.  Suppose that
  \begin{align*}
    ρ & = q_0 \trans{a_1} q_1 \trans{a_2} q_2 \trans{a_3} q_3 ⋯ \\
    ρ' & = q'_0 \trans{a_1} q'_1 \trans{a_2} q'_2 \trans{a_3} q'_3 ⋯ 
  \end{align*}
  where $q_0 = q'_0 = p$.  Let $n$ be the least integer such that
  $q_n ≠ q'_n$.  Let $a$ be the symbol~$a_n$, $w$ be the prefix
  $a_1 ⋯ a_{n-1}$ and $x'$ be the tail $a_{n+1}a_{n+2}a_{n+3}⋯$.  The
  sequence~$x$ is equal to $wax'$ and there is a finite run
  $q_0 \trans{w} q_{n-1}$, two transitions $q_{n-1} \trans{a} q_n$ and
  $q_{n-1} \trans{a} q'_n$, and the tail~$x'$ belongs to the intersection
  $\clofut{q_n} ∩ \clofut{q'_n}$.  We have actually proved the following
  equality expressing $\clobifut{p}$ in term of a union of intersections of
  sets~$\clofut{q}$.
  \begin{displaymath}
    \clobifut{p} = ⋃_{\begin{smallmatrix}
               p \trans{w} p' \\
               p' \trans{a} q \\
               p' \trans{a} q' 
             \end{smallmatrix}}
            wa(\clofut{q} ∩ \clofut{q'})
  \end{displaymath}
  Since the union is countable, it suffices to prove that
  if there are two transitions $p \trans{a} q$ and $p \trans{a} q'$ with
  $q ≠ q'$, then $μ(\clofut{q} ∩ \clofut{q'}) = 0$.
  Lemma~\ref{lem:ambiguity2} and Lemma~\ref{lem:add-prefix} allow us to
  conclude.
  
  The fact that (ii) implies (iii) is clear because the set~$\bifut{q}$ is
  contained in~$\clobifut{q}$ for each state~$q$ of~$𝒜$.
  
  We now prove that (iii) implies (i).  Let $q$ be state of~$𝒜$ and suppose
  that there are two different runs from state~$p$ to state~$r$ with the
  same label~$w$.  Let $v$ the label of a run from~$q$ to~$p$.  This shows
  that $vw\future{r} ⊆ \bifut{q}$.  Since $vw ∈ \past{r}$, the measure
  $μ(vw\future{r})$ satisfies $μ(vw\future{r}) = μ(vw\clofut{r})$ by
  Proposition~\ref{pro:same-measure}.  The measure $μ(vw\clofut{r})$
  satisfies $μ(vw\clofut{r}) > 0$ by Lemma~\ref{lem:positive-measure} and
  thus $μ(\bifut{q}) > 0$.  This completes the proof of this implication.
\end{proof}

\section*{Conclusion}

As a conclusion, we would like the emphasize the difficulty of extending
the result to non strongly connected automata.  Consider the two automata
pictured in Figure~\ref{fig:non-connected}.  The leftmost one in
unambiguous where as the rightmost one is obviously ambiguous for finite
words.  However, $\bifut{0} = \clobifut{0} = 0^*1^ℕ$ and $\bifut{q} = ∅$
hold for $q ≠ 0$ in both automata.  The only way to distinguish one
automaton from the other one is to have two different measures.  In order
to have $μ_1(\bifut{0}) = 0$ for the leftmost automaton, the measure~$μ_1$
should put all the weight on $0^ℕ$: $μ_1(0^ℕ) = 1$.  In order to have
$μ_2(\bifut{0}) > 0$ for the rightmost automaton, the measure~$μ_2$ should
put some weight on a sequence $0^n1^ℕ$ for some integer $n ⩾ 0$.  It is not
clear why the measures should be different because the set~$0^*1^*$ of
finite words labelling some run is the same in both automata.

\begin{figure}[htbp]
  \begin{center}
  \begin{tikzpicture}[>=stealth',initial text=,auto,inner sep=1pt]
    \tikzstyle{every state}=[semithick,minimum size=0.4]
    \begin{scope}
      \node[state,initial left] (q0) at (0,0.4) {$0$};
      \node[state,accepting] (q1) at (1.5,0.8) {$1$};
      \node[state,accepting] (q2) at (1.5,0) {$2$};
      \path[->] (q0) edge[out=120,in=60,loop] node {$0$} (q0);
      \path[->] (q0) edge node[pos=0.6] {$1$} (q1);
      \path[->] (q0) edge node[pos=0.6,swap] {$1$} (q2);
      \path[->] (q1) edge[out=30,in=-30,loop] node {$1$} (q1);
      \path[->] (q2) edge[out=30,in=-30,loop] node {$1$} (q2);
    \end{scope}
    \begin{scope}[xshift=150]
      \node[state,initial left] (q0) at (0,0.4) {$0$};
      \node[state,accepting] (q1) at (1.5,0.8) {$1$};
      \node[state,accepting] (q2) at (1.5,0) {$2$};
      \node[state,accepting] (q3) at (3,0.4) {$3$};
      \path[->] (q0) edge[out=120,in=60,loop] node {$0$} (q0);
      \path[->] (q0) edge node[pos=0.6] {$1$} (q1);
      \path[->] (q0) edge node[pos=0.6,swap] {$1$} (q2);
      \path[->] (q1) edge node[pos=0.4] {$1$} (q3);
      \path[->] (q2) edge node[pos=0.4,swap] {$1$} (q3);
      \path[->] (q3) edge[out=30,in=-30,loop] node {$1$} (q3);
    \end{scope}
  \end{tikzpicture}
  \end{center}
  \caption{Two non strongly connected automata}
  \label{fig:non-connected}
\end{figure}
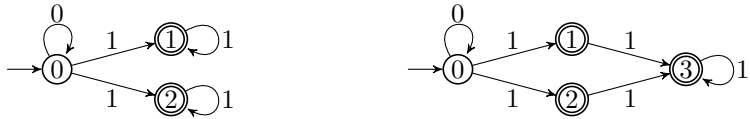

\bibliographystyle{plain}
\bibliography{measure}

\begin{thebibliography}{10}

\bibitem{AhoHopcroftUllman74}
A.~V. Aho, J.~E. Hopcroft, and J.~D. Ullman.
\newblock {\em The Design and Analysis of Computer Algorithms}.
\newblock Addison-Wesley, 1974.

\bibitem{BecherCarton18}
V.~Becher and O.~Carton.
\newblock Normal numbers and computer science.
\newblock In V.~Berth{\'e} and M.~Rigo, editors, {\em Sequences, Groups, and
  Number Theory}, Trends in Mathematics Series, pages 233--269. Birkh{\"a}user,
  2018.

\bibitem{BerstelPerrin84}
J.~Berstel and D.~Perrin.
\newblock {\em Theory of Codes}.
\newblock Academic Press, 1984.

\bibitem{BerstelReutenauer10}
J.~Berstel and Ch. Reutenauer.
\newblock {\em Noncommutative Rational Series with Applications}.
\newblock Cambridge University Press, 2010.

\bibitem{Borel09}
{\'E}.~Borel.
\newblock Les probabilit{\'e}s d{\'e}nombrables et leurs applications
  arithm{\'e}tiques.
\newblock {\em Rendiconti Circ. Mat. Palermo}, 27:247--271, 1909.

\bibitem{BousquetLoeding10}
N.~Bousquet and Ch. L{\"{o}}ding.
\newblock Equivalence and inclusion problem for strongly unambiguous
  {B}{\"{u}}chi automata.
\newblock In {\em {LATA} 2010}, volume 6031 of {\em Lecture Notes in Computer
  Science}, pages 118--129. Springer, 2010.

\bibitem{Carton22}
O.~Carton.
\newblock Preservation of normality by unambiguous transducers.
\newblock {\em CoRR}, abs/2006.00891, 2022.

\bibitem{CartonMichel03}
O.~Carton and M.~Michel.
\newblock Unambiguous {B\"u}chi automata.
\newblock {\em Theoret. Comput. Sci.}, 297:37--81, 2003.

\bibitem{ChoffrutGrigorieff99}
Ch. Choffrut and S.~Grigorieff.
\newblock Uniformization of rational relations.
\newblock In {\em Jewels are Forever, Contributions on Theoretical Computer
  Science in Honor of Arto Salomaa}, pages 59--71. Springer, 1999.

\bibitem{Colcombet15}
Th. Colcombet.
\newblock Unambiguity in automata theory.
\newblock In {\em Descriptional Complexity of Formal Systems}, volume 9118 of
  {\em Lecture Notes in Computer Science}, pages 3--18. Springer, 2015.

\bibitem{Halmos63}
P.~R. Halmos.
\newblock {\em Lectures on Boolean algebras}.
\newblock Von Nostrand, 1963.

\bibitem{HanselPerrin90}
G.~Hansel and D.~Perrin.
\newblock Mesures de probabilit{\'e} rationnelles.
\newblock In M.~Lothaire, editor, {\em Mots}, pages 335--357. Hermes, 1990.

\bibitem{HolzerKutrib10}
M.~Holzer and M.~Kutrib.
\newblock Descriptional complexity of (un)ambiguous finite state machines and
  pushdown automata.
\newblock In {\em Reachability Problems}, pages 1--23. Springer, 2010.

\bibitem{HopcroftUllman79}
J.~E. Hopcroft and J.~D. Ullman.
\newblock {\em Introduction to Automata Theory, Languages, and Computation}.
\newblock Addison-Wesley, 1979.

\bibitem{IsaakLoeding12}
D.~Isaak and Ch. L{\"{o}}ding.
\newblock Efficient inclusion testing for simple classes of unambiguous
  {\(\omega\)}-automata.
\newblock {\em Inf. Process. Lett.}, 112(14-15):578--582, 2012.

\bibitem{LindMarcus95}
D.~Lind and B.~Marcus.
\newblock {\em An Introduction to Symbolic Dynamics and Coding}.
\newblock Cambridge University Press, 1995.

\bibitem{PerrinPin04}
D.~Perrin and J.-{\'E}. Pin.
\newblock {\em Infinite Words}.
\newblock Elsevier, 2004.

\bibitem{Pighizzini12}
G.~Pighizzini.
\newblock Two-way finite automata: Old and recent results.
\newblock {\em Electronic Proceedings in Theoretical Computer Science},
  90:3–20, 2012.

\bibitem{Sakarovitch09a}
J.~Sakarovitch.
\newblock {\em Elements of Automata Theory}.
\newblock Cambridge University Press, 2009.

\bibitem{Schmidt78}
E.~Schmidt.
\newblock Succinctness of descriptions of context-free, regular, and finite
  languages.
\newblock {\em DAIMI Report Series}, 7(84), 1978.

\bibitem{StearnsHunt85}
R.~E. Stearns and H.~B.~Hunt III.
\newblock On the equivalence and containment problems for unambiguous regular
  expressions, regular grammars and finite automata.
\newblock {\em {SIAM} J. Comput.}, 14(3):598--611, 1985.

\end{thebibliography}

\end{document}